\documentclass{article}

\usepackage{amsmath,mathtools}
\usepackage{algorithm}
\usepackage{algcompatible}
\usepackage{graphicx,psfrag,amsmath,amssymb,epsf}
\usepackage{epsfig,graphics}

\newtheorem{theorem}{Theorem}[section]
\newtheorem{lemma}{Lemma}[section]
\newtheorem{corollary}{Corollary}[section]

\newcommand{\qed}{\hfill\hbox{\rlap{$\sqcap$}$\sqcup$}}
\newenvironment{proof}{\noindent \emph{Proof.\,}}{\qed}

\algnewcommand\algorithmicreturn{\textbf{return}}
\algnewcommand\RETURN{\algorithmicreturn}
\algnewcommand\algorithmicprocedure{\textbf{procedure}}
\algnewcommand\PROCEDURE{\item[\algorithmicprocedure]}%
\algnewcommand\algorithmicendprocedure{\textbf{end procedure}}
\algnewcommand\ENDPROCEDURE{\item[\algorithmicendprocedure]}%
\algnewcommand{\algvar}[1]{{\text{\ttfamily\detokenize{#1}}}}
\algnewcommand{\algarg}[1]{{\text{\ttfamily\itshape\detokenize{#1}}}}
\algnewcommand{\algproc}[1]{{\text{\ttfamily\detokenize{#1}}}}
\algnewcommand{\algassign}{\leftarrow}

\title{Optimal Bridge, Twin Bridges and Beyond: Inserting Edges into a Road Network to Minimize the Constrained Diameters} 

\author{Zhidan Feng\footnote{University of Trier, FB 4 – Informatikwissenschaften, 54286, Trier, Germany. Email: {\tt s4zhfeng@uni-trier.de}.}
\and
Henning Fernau\footnote{University of Trier, FB 4 – Informatikwissenschaften, 54286, Trier, Germany. Email: {\tt fernau@uni-trier.de}.} 
\and
Binhai Zhu\footnote{Gianforte School of Computing, Montana State University, Bozeman, MT 59717, USA. Email: {\tt bhz@montana.edu}.}
}

\date{}

\begin{document}

\maketitle

\begin{abstract}
Given a road network modelled as a planar straight-line graph $G=(V,E)$ with $|V|=n$, let $(u,v)\in V\times V$, the shortest path (distance) between $u,v$ is denoted as $\delta_G(u,v)$. Let $\delta(G)=\max_{(u,v)}\delta_G(u,v)$, for $(u,v)\in V\times V$, which is called the diameter of $G$. Given a disconnected road network modelled as two disjoint trees $T_1$ and $T_2$, this paper first aims at inserting one and two edges (bridges) between them to minimize the (constrained) diameter $\delta(T_1\cup T_2\cup I_j)$ going through the inserted edges, where $I_j, j=1,2$, is the set of inserted edges with $|I_1|=1$ and $|I_2|=2$. The corresponding problems are called the 
{\em optimal bridge} and {\em twin bridges} problems. Since when more than one edge are inserted between two trees the resulting graph is becoming more complex, for the general network $G$ we consider the problem of inserting a minimum of $k$ edges such that the shortest distances between a set of $m$ pairs 
$P=\{(u_i,v_i)\mid u_i,v_i\in V, i\in [m]\}$, $\delta_G(u_i,v_i)$'s, are all decreased.

The main results of this paper are summarized as follows:
\begin{itemize}
    \item We show that the optimal bridge problem can be solved in $O(n^2)$ time and that a variation of it has a near-quadratic lower bound unless SETH fails. The proof also implies that the famous 3-SUM problem does have a near-quadratic lower bound for large integers, e.g., each of the $n$ input integers has 
    $\Omega(\log n)$ decimal digits.
    We then give a simple factor-2 $O(n\log n)$ time approximation algorithm for the optimal bridge problem.
    \item We present an $O(n^4)$ time algorithm to solve the twin bridges problem, exploiting some new property not in the optimal bridge problem.
    \item For the general problem of inserting $k$ edges to reduce the (graph) distances between $m$ given pairs, we show that the problem is NP-complete --- even if the given and resulting graphs are planar and the inserted edges must have a bounded length.
\end{itemize} 
\end{abstract}

\section{Introduction}

Geometric spanners have received a lot of attention since 1986 \cite{DBLP:conf/compgeom/Chew86}. In the majority of the literature, the problem is to construct a geometric (or metric) graph $G$ from scratch such
that the maximum stretch factor (or dilation) of $G$ is minimized (or approximated). Recently, a breakthrough result of Gudmundsson and Wong proved that by inserting $k$ edges greedily into a connected metric graph $M$ the dilation of the resulting graph can be approximated with a factor of $O(k)$ \cite{GW22}.

In this paper, we consider a similar problem of inserting edges to an existing road network $G$ (modelled as a planar straight-line graph of $n$ vertices) 
such that the diameter of the resulting graph is minimized/reduced. We start from the basic problem of inserting one and two edges into two disjoint trees~$T_1$ and~$T_2$ to minimize the constrained diameter, going through the inserted edges, of the resulting graphs, which are called the {\em optimal bridge} and {\em twin bridges} problems. These problems could appear in some scenario, for example, in the recent work by Higashikawa et al., where the vertices of $G$ could be in red and black and a path must be planned with no red-red edge (i.e., no edges whose vertices are both red)~\cite{COCOON23}. Naturally, such a path could be planned by deleting all red-red edges in $G$,
possibly resulting in disjoint components (disjoint trees or a forest when $G$ is a tree). Then, to ensure connectivity and minimizing/reducing the dilation some 
non-red-red edges much be inserted --- even though how to do that is left open.

Coming back to the optimal bridge and twin bridges problems, we present $O(n^2)$ and $O(n^4)$ time algorithms to solve them. For the former, we show that a variation of it, the one-bridge decision problem, has a near-quadratic lower bound unless the Strong Exponential Time Hypothesis (SETH) fails. Our method 
modifies the traditional one by Williams and it can be used to prove that the famous 3-SUM problem has a near-quadratic lower bound when the $n$ input integers are large, i.e., each has $\Omega(\log n)$ decimal digits.
Then, a simple
$O(n\log n)$ time approximation algorithm is designed to achieve a factor-2 (the factor is tight). Consequently, this approximation algorithm can be
used as a subroutine to approximate the more general problem of connecting $k+1$ trees with $k$ edges into a tree~$T$ so as to minimize the diameter of~$T$, with a factor of 4 and a running time of $O(n\log n)$. For the twin bridges problem, not surprisingly, when two edges are added between~$T_1$ and $T_2$, the resulting graph is not a tree anymore. Hence we need to consider a new scenario that never occurs in the optimal bridge problem.

On the other hand, from a pure geometric setting, the optimal bridge problem is well-studied when the input is a pair of polygons.
The problem was first studied by Cai, Xu and Zhu in 1999 for the case where the input is a pair of convex polygons~$P$ and~$Q$, with a total of $n$ vertices \cite{DBLP:journals/ipl/CaiXZ99}. (In this case, the problem is to construct a connecting  segment $pq$ between $P,Q$ such that $p\in P$ and $q\in Q$ and
$\max_{x\in P}|xp|+|pq|+\max_{y\in Q}|qy|$ is minimized. Here, $|xp|$ is the Euclidean distance between points $x$ and $p$.)
The problem was solved in $O(n^2\log n)$ time and was subsequently improved to $O(n)$ by using geometric properties \cite{DBLP:journals/ipl/Tan00,DBLP:journals/ipl/BhattacharyaB01,DBLP:journals/mst/KimS01}. When the polygons $P,Q$ are not necessarily convex, the problem is 
to find a connecting segment $pq$ between $P,Q$ such that $p\in P$ and $q\in Q$ to minimize
$\max_{x\in P}\delta_P(x,p)+|pq|+\max_{y\in Q}\delta_Q(q,y)$. (Here $\delta_P(x,p)$ is the geodesic distance between $x$ and $p$ in polygon $P$.)
For this problem, there are two versions: (1) $p,q$ are visible to each other, and (2) $p,q$ do not have to be visible to each other.
For the first case, the problem can be solved in $O(n^2)$ time by Kim and Shin~\cite{DBLP:journals/mst/KimS01}, improved to $O(n\log^3 n)$ time by Tan~\cite{DBLP:journals/ijcga/Tan02}. For the second case, Bhosle and Gonzalez gave an $O(n^2\log n)$ time algorithm \cite{DBLP:journals/ijcga/BhosleG05} (and they solved a variation when $p,q$ could be connected by one or more segments in the same amount of time~\cite{DBLP:conf/cccg/BhosleG11}).
When $P,Q$ are rectilinear and $p,q$ are connected by one or two rectilinear segments, Wang gave an optimal $O(n)$ time algorithm \cite{DBLP:journals/ipl/Wang01}.

Also in the geometric setting, when $P,Q$ are both convex, a simple $O(n)$ time greedy 2-approximation algorithm was given by Cai, Xu and Zhu \cite{DBLP:journals/ipl/CaiXZ99}. The algorithm is simply to connect two closest points between $P$ and $Q$. Slightly later, Ahn, Cheong and Shin gave
a factor-$\sqrt{2}$ approximation by connecting the centers of $P$ and $Q$ (i.e., the centers of the circumscribing circles of $P$ and $Q$), which can be
used to obtain a factor-$2\sqrt{2}$ approximation to connect $k+1$ convex polygons by inserting $k$ edges to minimizing the (geodesic) diameter. In our case, 
we show that the greedy algorithm by Cai, Xu and Zhu can be adapted to the road network setting to achieve a factor-2 approximation and the factor is tight.
The algorithm can also be applied to obtain a factor-4 approximation to connect $k+1$ trees into a tree~$T$ by inserting $k$ edges so as to minimize the diameter of~$T$.

As discussed earlier, when two or more edges are added between~$T_1$ and~$T_2$, the resulting graph is not a tree anymore. Hence, algorithms must make use of some new properties. We prove such properties and present an $O(n^4)$ time algorithm to solve the problem.

On the other hand, it is noted that when two or more edges are inserted to connect $T_1$ and $T_2$ it might be possible that it cannot reduce the diameter further --- regardless of if the second edge is inserted between $T_1$ and $T_2$ or between the vertices in $T_i, i\in \{1,2\}$.
(An example is when $T_1$ and $T_2$ are two disjoint segments on the same line.) Hence, the problem here is more on inserting edges to a connected graph to minimize/reduce the diameter, closer to the setting by Gudmundsson and Wong~\cite{GW22}.

Finally, we investigate the complexity of the following problem: given a connected road network $G=(V,E)$ and a set of pairs $P=\{(u,v)\mid u,v\in V\}$, insert $k$ edges such that the distances $\delta_G(u,v)$, with $(u,v)\in P$, can all be reduced.
We prove that this problem is NP-complete by a reduction from
Vertex Cover on Planar 2-Connected Cubic Graphs \cite{BM01}.

The paper is organized as follows. In Section 2 we give necessary definitions. In Section 3 we study the optimal bridge problem. In Section 4 we present the results for the twin bridges problem. In Section 5, we present the NP-completeness proof for the general case. We conclude the paper in Section 6.

\section{Preliminaries}

We define a road network as a (weighted) planar straight-line graph
$G=(V,E)$. Each vertex (node) $p\in V$ is a point in the plane
with coordinates $p=(x(p),y(p))$. The weight between two points $p,q$ is their Euclidean distance 
which is denoted as $|pq|$, and is equal to $\sqrt{(x(p)-x(q))^2+(y(p)-y(q))^2}$. Given $u,v\in V$, the shortest path between them on $G$ is denoted as
$\delta_G(u,v)$. (We also slightly abuse the notation to use $\delta_G(u,v)$ as the length of the corresponding shortest path.) The diameter of $G$ is defined as $\max_{u,v\in V}\delta_G(u,v)$.

A tree $T=(V(T),E(T))$ is a connected road network with no cycle. As well-known facts, between $u,v\in V(T)$ the shortest path $\delta_T(u,v)$ is unique.  Also, if $|V(T)|=n$, then the shortest path tree from any node $u$ to all other nodes can be computed in $O(n)$ time by running the breadth-first search (BFS) algorithm starting at $u$. Consequently, the all-pair shortest paths problem on~$T$ can be solved in $O(n^2)$ time; moreover, when this information is stored as a table, the length of $\delta_T(u,v)$ can be returned in $O(1)$ time. In fact, $\max_{x\in V(T)}\delta_T(u,x)$, the maximum distance from $u\in V(T)$ to a (leaf) node $x$, can be computed in $O(n)$ time by looking at the row/column maximum; moreover, once stored, this distance can be queried in $O(1)$ time. Throughout the paper, we assume such a table on all-pair shortest path distances is associated with the corresponding tree.

\section{The optimal bridge problem}

The optimal bridge problem is defined as follows: given two tree road networks (trees for short) $T_1$ and~$T_2$, add an edge $(p,q)$, with $p\in V(T_1)$, $q\in V(T_2)$, such that the (constrained) diameter through $pq$ of the resulting graph $T_1\cup T_2\cup \{(p,q)\}$, i.e.,
the maximum of $\delta_{T_1}(x,p)+|pq|+\delta_{T_2}(q,y)$ with $x\in V(T_1)$ and $y\in V(T_2)$,
is minimized.

An $O(n^2)$ time algorithm can be easily designed to solve the optimal bridge problem: 
\begin{enumerate}
    \item For each pair, $p\in V(T_1)$, $q\in V(T_2)$, record
    $\max_{x\in V(T_1)}\delta_{T_1}(x,p)+|pq|+\max_{y\in V(T_2)}\delta_{T_2}(q,y)$.
    \item Return the minimum distance recorded and then compute the corresponding path between~$x$ and~$y$.    
\end{enumerate}

It is intriguing to know if this $O(n^2)$ time bound can be further improved. For the (pure) geometric version where two polygons with vertices $n$ are given and the bridge must be a segment $pq$ where $p$ is visible from $q$, an $O(n\log^3 n)$ time algorithm is known \cite{DBLP:journals/ijcga/Tan02}.

We define the {\em one-bridge} problem as the decision version of the optimal bridge: given constants $C_1$ and $C_2$, determine if a bridge $pq$ exists such that $|pq|= C_1$ and its {\em specific} solution value $\delta_{T_1}(x,p)+|pq|+\delta_{T_2}(q,y) = C_2$,
where $x$ is a leaf in $T_1$ and $y$ is a leaf in $T_2$.
Clearly, one-bridge is a variation of the optimal bridge problem. We show next that the one-bridge problem cannot be solved in $O(n^{2-\epsilon})$ time unless the SETH (Strong Exponential Time Hypothesis) fails. Thus, even though this does not give a similar lower bound for the optimal bridge problem, it does give some evidence that it is probably hard to
improve the $O(n^2)$-time bound for the optimal bridge problem. 

\subsection{A near-quadratic lower bound for the one-bridge problem}

We prove the conditional lower bound by considering a variation of Williams' Orthogonal Vector (OV) problem
\cite{DBLP:journals/tcs/Williams05}, which we call {\em Complementary Orthogonal Vectors} problem (COV for short).

Given a binary (0/1) vector in $d$-dimension, $v=(v_1,v_2,\dots,v_d)$, define $\bar{v}=(\bar{v}_1,\bar{v}_2,\dots,\bar{v}_d)$ (with $\bar{0}=1$ and $\bar{1}=0$). We say that $v$ and $\bar{v}$ are {\em complementary} to each other. Clearly, $v$ and $\bar{v}$ are orthogonal, i.e., the dot product satisfies $v\cdot \bar{v}=0$.  Roughly, in Williams' case, we given a $k$-SAT instance $\phi'$ with a set of $n$ variables $A\cup B$, with $|A|$ and $|B|$ being of roughly the same size $n/2$.  Then two sets of binary vectors in $d$-dimension, $A'$ and~$B'$, which correspond to how a partial assignment 
of $A$ and $B$ satisfy the $d$ clauses in $\phi'$, are constructed. Here the $j$-th component of a $d$-vector $u\in A'$ is 0 if the partial assignment of $A$ can satisfy the $j$-th clause in
$\phi'$, and 1 otherwise. (A vector in $B'$ is similarly defined.) Then, Williams' problem is to find if there are binary vectors $u\in A'$ and $v\in B'$ such that they are orthogonal, i.e., $u\cdot v=0$. In our {\em Complementary Orthogonal Vectors} (COV) problem, we have two sets of {\em ternary} vectors $A$ and $B$ such that each entry in a vector of $A$ and $B$ is an element of $\{0,1,2\}$, and we want to find a pair of binary vectors $u\in A$ and $v\in B$ such
that they are complementary.

To modify William's proof, we make the following changes: (1) Use One-in-three SAT instead of $k$-SAT;
(2) Modify the definition for constructing the sets of $d$-vectors.

Let $\phi$ be an One-in-three SAT instance composed of $n$ variables and $m$ disjunctive clauses where the $i$-th clause $F_i$ contains three literals and is in the form of $(x_{i,1}\vee x_{i,2}\vee x_{i,3})$. The problem is to determine for $i=1..m$, exactly one of the three literals in each clause $F_i$, i.e., $x_{i,1}, x_{i,2}$ and $x_{i,3}$, is assigned {\tt TRUE}. One-in-three SAT is a well-known NP-complete problem with $m=\Omega(n)$ \cite{DBLP:conf/stoc/Schaefer78}.

We arbitrarily partition the variables in $\phi$ into two equal-sized parts $V_A$ and $V_B$ (we can assume that $n$ is even, though it does not really matter for the result). Each of the $m$-vectors $u\in A$ is determined by an assignment $\alpha_A$ of $V_A$ (i.e., a partial assignment of the variables in $\phi$), where
$$u=(u_1,u_2,\cdots,u_m),$$
and 
\begin{equation}
    u_i=
    \begin{cases}
        0 & \text{if}~F_i ~\text{is satisfied with exactly one TRUE literal by}~ \alpha_A,  \\
        1 & \text{if}~F_i ~\text{is not satisfied by}~ \alpha_A,\\
        2 & \text{if}~F_i ~\text{is satisfied with at least two TRUE literals by}~ \alpha_A.
    \end{cases}
\end{equation}
Similar to (1), we could define
an $m$-vector $v\in B$ determined by an assignment $\alpha_B$ of $V_B$.
In a preprocessing, we could remove all non-binary vectors from $A$ and $B$, but this does not affect the lower bound proof (for the worst case). Hence we will stick with the ternary vectors as input
for COV.

Then, similar to Williams' idea, we can claim that $\phi$ has
a valid truth assignment if and only if there are vectors $u\in A$ and $v\in B$ which are complementary (i.e., COV has a solution). As there are $2^{n/2}$ assignments for $V_A$ and $V_B$ respectively, the above reduction takes $2^{n/2}\cdot O(m)$ time. If COV could be computed in $O(N^{2-\epsilon})$ time, where $N$ is the input size for COV, One-in-three SAT could be solved in $2^{n-n\epsilon/2}\cdot O(m^{2-\epsilon})$ time --- which would contract the SETH. We hence have the following theorem.

\begin{theorem}
The Complementary Orthogonal Vectors problem with input size~$N$ cannot be solved in $O(N^{2-\epsilon})$ unless the SETH
fails.
\label{thm1}
\end{theorem}

We next reduce COV to one-bridge as follows> Let the input for COV be two sets $A$ and $B$ each containing $n_1$ and $n_2$ (0/1/2)-vectors in $m$-dimension respectively, with $n_1+n_2=n$,
Hence the total size of COV is $N=mn$. For each $m$-vector $u\in A$, with $u=(u_1,u_2,\cdots, u_m)$,
we construct a path $L_u$ of segment lengths $\ell_i$, i.e., $L_u=(\ell_1,\ell_2,\cdots, \ell_m)$,
where
\begin{equation}
    \ell_i=
    \begin{cases}
        \frac{1}{3^{i-1}} & \text{if} ~u_i=0,\\
        0 & \text{if}~ u_i=1,\\
        4 & \text{if}~ u_i=2.
    \end{cases}
\end{equation}
Here $\ell_i$ is the length of the $i$-th segment of $L_u$. (The fact that $\ell_i=0$ means that we could have duplicated
geometric points on $L_u$.)
Similarly, given an $m$-vector $v\in B$, with $v=(v_1,v_2,\cdots,v_m)$, we can construct a path $L_v$ of segment lengths $\ell'_i$, i.e., $L_v=(\ell'_1,\ell'_2,\cdots,\ell'_m)$, with
\begin{equation}
    \ell'_i=
    \begin{cases}
        \frac{1}{3^{i-1}} & \text{if} ~v_i=0,\\
        0 & \text{if}~ v_i=1, \\
        4 & \text{if}~ v_i=2.
    \end{cases}
\end{equation}

Let $C_1=0$ and $C_2=C+2$ with
$$C=1+\left(\frac{1}{3}\right)+\left(\frac{1}{3}\right)^2+\cdots+\left(\frac{1}{3}\right)^{m-1}=\frac{3}{2}\left(1-\left(\frac{1}{3}\right)^m\right).$$
Then we put these straight paths $L_u$'s not containing segment lengths of four downward along the $Y$-axis starting at $(0,C)$. We also add another node $(0,C+1)$ and connect it to $(0,C)$. For $L_u$'s which contain a segment of length 4, we just convert each of them into a path centered at $(0,C)$, in a star fashion.
This would give us $T_1$.
Similarly, we define put the straight path $L_v$'s not containing a segment of length 4, each corresponding to
a binary vector $v\in B$, upward and starting at $(0,0)$ and we also add another node $(0,-1)$ and connect it to $(0,0)$. For $L_v$'s containing a segment of length 4, we put them in a star centered at $(0,0)$. This gives us the second tree $T_2$. See Figure~\ref{fig:fig0} for an example.

\begin{figure}
\psfrag{p,q}{$p,q$}
\psfrag{T1}{$T_1$}
\psfrag{T2}{$T_2$}
\psfrag{Lu}{$L_u$}
\psfrag{Lv}{$L_v$}
\centering
\includegraphics[width=0.45\textwidth]{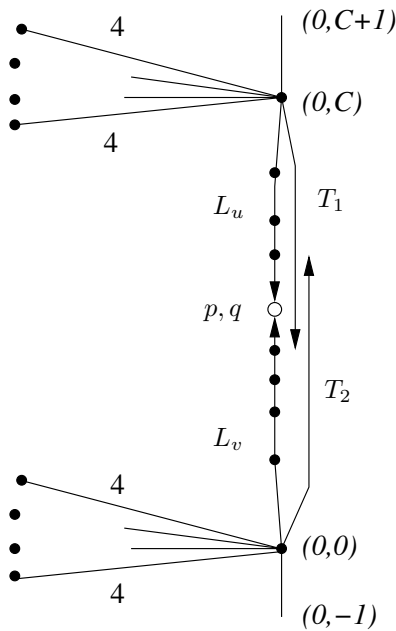}
\caption{\label{fig:fig0}An example for the reduction from COV to the one-bridge decision problem. The vertical paths, which should all be on the $Y$-axis, are drawn for a better visualization.}
\end{figure}

This reduction obviously takes $O(N)=O(mn)$ time, where the total number of points in $T_1$ and $T_2$ is $O(N)$. Finally, we claim that COV has a solution if and only if between $T_1$ and $T_2$ there is a bridge $pq$ with length $C_1=0$ and with
a specific solution value of $C_2=C+2$. We prove this ``iff'' relation next, focusing on the ``if-part'' as the ``only-if''
part is trivial. We first prove the following lemma regarding the geometric sequence $(1,\frac{1}{3},\frac{1}{3^2},\cdots,\frac{1}{3^{m-1}})$. Note that as $C<\frac{3}{2}$,
$C_2<3.5$. Hence, we can ignore any path with a segment length
of 4 in our proof, as they would incur a specific solution
of value at least 4, which is greater than $C_2$.



\begin{lemma}
For $0\leq i< m-1$, $\left(\frac{1}{3}\right)^i>\sum_{j=i+1..m-1}\left(\frac{1}{3}\right)^{j}$.
\label{lem1}
\end{lemma}

\begin{proof}
By calculation, $$\sum_{j=i+1..m-1}\left(\frac{1}{3}\right)^{j}=\left(\frac{1}{3}\right)^{i+1}\cdot \frac{3}{2}\left(1-\left(\frac{1}{3}\right)^{m-i-1}\right)<\frac{1}{2}\cdot \left(\frac{1}{3}\right)^i< \left(\frac{1}{3}\right)^i.$$
\end{proof}

We then prove the ``if-part'' by proving its contrapositive: if COV does not have a solution, then
between $T_1$ and $T_2$ there does not exist a bridge $pq$ with length $C_1=0$ and with
a specific solution value of $C_2=C+2$. To finish this proof, we need the following two lemmas.

\begin{lemma}
 Let $u\in A, v\in B$; moreover, $u=(u_1,u_2,\cdots,u_m)$ and
 $v=(v_1,v_2$, $\cdots,v_m)$ with $u_i,v_i\in\{0,1\}$. If there exists an $i$ with $u_i=v_i=0$ and for all $k<i$ exactly one of $u_k$ and $v_k$ is equal to zero, then 
 the sum of segment lengths in $L_u$ and $L_v$ is greater than $C$. 
 \label{lem2}
\end{lemma}

\begin{proof}
 Here $i$ is the smallest index such that $u_i=v_i=0$ and for all $k<i$ exactly one of $u_k$ and $v_k$ is equal to zero. Then, by definition, both $\ell_i=\left(\frac{1}{3}\right)^{i-1}$ and $\ell'_i=\left(\frac{1}{3}\right)^{i-1}$ would contribute to the sum of segment lengths in $L_u$ and $L_v$. Following Lemma~\ref{lem1},
 $$\left(\frac{1}{3}\right)^{i-1}>\sum_{j=i..m-1}\left(\frac{1}{3}\right)^{j}.$$
 In other words, the value of this additional copy of
 $\left(\frac{1}{3}\right)^{i-1}$ is greater than the sum of all possible shorter segment lengths, each appearing exactly once, which is necessary to achieve a sum of $C$.
 Hence there is no bridge $pq$ with a length of zero connecting $T_1$ and $T_2$, and the lemma is proven.
\end{proof}

\begin{lemma}
 Let $u\in A, v\in B$; moreover, $u=(u_1,u_2,\cdots,u_m)$ and
 $v=(v_1,v_2$, $\cdots,v_m)$ with $u_i,v_i\in\{0,1\}$. If there exists an $i$ with $u_i=v_i=1$ and for all $k<i$ exactly one of $u_k$ and $v_k$ is equal to zero, then 
 the sum of segment lengths in $L_u$ and $L_v$ is smaller than $C$. 
 \label{lem3}
\end{lemma}

\begin{proof}
 Symmetrically, $i$ is the smallest index such that $u_i=v_i=1$ and for all $k<i$ exactly one of $u_k$ and $v_k$ is equal to zero. Again, by definition, both $\ell_i=\left(\frac{1}{3}\right)^{i-1}$ and $\ell'_i=\left(\frac{1}{3}\right)^{i-1}$ would not contribute to the sum of segment lengths in $L_u$ and $L_v$. Again, by the proof of Lemma~\ref{lem1},
 $$\frac{1}{2}\cdot\left(\frac{1}{3}\right)^{i-1}>\sum_{j=i..m-1}\left(\frac{1}{3}\right)^{j},$$ or
 $$\left(\frac{1}{3}\right)^{i-1}>2\cdot\sum_{j=i..m-1}\left(\frac{1}{3}\right)^{j}.$$ 
 In other words, for all possible segment lengths shorter than $\left(\frac{1}{3}\right)^{i-1}$, each possibly appearing twice in $L_u$ as well as in $L_v$, their sum would not make up for $\left(\frac{1}{3}\right)^{i-1}$ which is needed to have
 a sum of $C$ for the segments in $L_u$ and $L_v$.
 Again, there cannot be a bridge $pq$ with a length of zero
 connecting $T_1$ and $T_2$.
\end{proof}

\begin{theorem}
Given two trees with a total of $N$ vertices and non-negative constants $C_1$ and $C_2$, the one-bridge decision problem cannot be solved in $O(N^{2-\epsilon})$ time unless the SETH fails.
\label{thm2}
\end{theorem}

\begin{proof}
As claimed earlier, COV has a solution if and only if between $T_1$ and $T_2$ there is a bridge $pq$ with length $C_1=0$ and with
a specific solution value of $C_2=C+2$. We finish to  prove this ``iff'' relation, focusing on the ``if-part'' as the ``only-if'' is trivial.

Suppose that there is a bridge $pq$ between $T_1$ and $T_2$ with length $C_1=0$ and with
a specific solution value of $C_2=C+2$. Since $C<1.5$, we have $C+2=C_2<3.5$. Hence the only possible specific solution value $C_2$ in the one-bridge instance is the (tree) distance between $(0,C+1)$ and $(0,-1)$, as all non-vertical segments have length 4. As $C_1=0$ (i.e., there is a bridge of length zero between $T_1$ and $T_2$), the sum of segment lengths in $L_u$ and $L_v$, which are paths in $T_1$ and $T_2$ respectively, must be exactly $C$. Following Lemma~\ref{lem2} and Lemma~\ref{lem3}, for that to happen, each component in the geometric sequence $(1,\frac{1}{3},\frac{1}{3^2},\cdots,\frac{1}{3^{m-1}})$ must appear exactly once in $L_u$ and $L_v$. Recall that the total number of points in $T_1$ and $T_2$ is $O(N)$. Then, if one-bridge could be solved in $O(N^{2-\epsilon})$ time the Complementary Orthogonal Vectors problem would also be solved in $O(N^{2-\epsilon})$ time. But this is a contradiction to Theorem~\ref{thm1}. 
\end{proof}

\subsection{Implication to the 3-SUM Conjecture}

The 3-SUM problem is defined as follows: given a set $S$ of $n$ integers, decide if there are
$a,b,c\in S$ such that $a+b+c=0$. The problem was initially
posed by Gajentaan and Overmars \cite{DBLP:journals/comgeo/GajentaanO95}, who conjectured that the problem has a lower bound of $\Omega(n^2)$. Then the problem was solved in $o(n^2)$ time by Baran et al. \cite{DBLP:journals/algorithmica/BaranDP08} (later even for real numbers \cite{DBLP:conf/focs/JorgensenP14}). Hence the new conjecture for 3-SUM is by Patrascu, which is that 3-SUM has a lower bound of $\Omega(n^{2-o(1)})$ \cite{DBLP:conf/stoc/Patrascu10}. We show that when the integers in $S$ are large, i.e., each has $\Omega(\log n)$ decimal digits, then the 3-SUM conjecture is true (in fact, it is almost matches the original conjecture by Gajentaan and Overmars). 

We follow the set-up for Theorem~\ref{thm1}.
Recall that $\phi$ is an One-in-three SAT instance composed of $n$ variables and $m$ disjunctive clauses where the $i$-th clause $F_i$ contains three literals and is in the form of $(x_{i,1}\vee x_{i,2}\vee x_{i,3})$. The problem is to determine for $i=1..m$, exactly one of the three literals in each clause $F_i$, i.e., $x_{i,1}, x_{i,2}$ and $x_{i,3}$, is assigned {\tt TRUE}.

We arbitrarily partition the variables in $\phi$ into two equal parts $V_A$ and $V_B$ (we can assume that $n$ is even, though it does not really matter for the result). Each of the $(m+2)$-vectors $u\in A$ is determined by an assignment $\alpha_A$ of $V_A$ (i.e., a partial assignment of the variables in $\phi$), where
$$u=(u_1,u_2,\cdots,u_{m+2}),$$
and, for $1\leq i\leq m$, 
\begin{equation}
    u_i=
    \begin{cases}
        0 & \text{if}~F_i ~\text{is not satisfied by}~ \alpha_A,\\
        1 & \text{if}~F_i ~\text{is satisfied with exactly one TRUE literal by}~ \alpha_A, \\
        2 & \text{if}~F_i ~\text{is satisfied with at least two TRUE literals by}~ \alpha_A.
    \end{cases}
\end{equation}
For $i=m+1$ and $i=m+2$, we define
\begin{equation}
    u_{m+1}=
    \begin{cases}
        1 & \text{if}~u~\text{is in}~ V_A,\\
        0 & \text{otherwise}.
    \end{cases}
\end{equation}

\begin{equation}
    u_{m+2}=
    \begin{cases}
        1 & \text{if}~v~\text{is in}~ V_B,\\
        0 & \text{otherwise}.
    \end{cases}
\end{equation}
Similar to (1), we could define
an $(m+2)$-vector $v\in B$ determined by an assignment $\alpha_B$ of $V_B$.
We define the integer $v^*(m+2)=\overbrace{11\cdots 1}^{m+2}$ (which can still be viewed as a vector
of $m+2$ ones).

Let $S$ be a set of $2^{n/2+1}+1$ large integers, each with $m+2$ decimal digits; more precisely, let $S=A\cup B\cup \{-v^*(m+2)\}$.
Then, similar to Williams' idea, we can claim that $\phi$ has
a valid truth assignment if and only if there are vectors (integers) $u\in A$ and $v\in B$ such that $u+v+(-v^*(m+2))=0$ (or, equivalently, if there are $u,v\in S$ such that $u+v+(-v^*(m+2))=0$, i.e., 3-SUM has a solution). As there are $2^{n/2}$ assignments for $V_A$ and $V_B$ respectively, the above reduction takes $2^{n/2}\cdot O(m)$ time. If 3-SUM could be computed in $O(N^{2-\epsilon})$ time, where $N$ is the number of input integers for 3-SUM, One-in-three SAT could be solved in $2^{n-n\epsilon/2}\cdot O(m^{2-\epsilon})$ time --- which would fail the SETH. (Note that when $N$ large input integers are given for 3-SUM, each of them needs to have $m+2=\Omega(\log N)$ decimal digits.) We hence have the following theorem.

\begin{theorem}
The 3-SUM problem with ~$N$ large integers, each with $\Omega(\log N)$ decimal digits, cannot be solved in $O(N^{2-\epsilon})$ time unless the SETH
fails.
\end{theorem}

We could easily extend the above ideas to $k$-SUM: just partition the variables in $\phi$ into $k$ subsets $V_i, 1\leq i\leq k$, each containing at most $n/k$ variables.
Then from each $V_i, 1\leq i\leq k$, construct a set of integers with 
$m+k$ decimal digits where the first $m$ digits are defined as in equation (1) while the $(m+j)$-th digit is one for $j=i$, and zero for all other $1\leq j\leq k$. However, note that, among the first $m$ decimal digits,
it is possible to obtain a ternary sub-vector of 1's due to possible carries. For instance, we could obtain a 2-vector (1,1) by adding six ternary 2-vectors, $(0,2)+(0,2)+(0,2)+(0,2)+(0,2)+(0,1)$, due to a carry from adding five 2's. Hence,
we need $k\leq 6$ to enforce the relation that $\phi$ has a valid truth assignment if and only if there are $v_i\in V_i$ such that $v_1+v_2+\cdots+v_{k-1}+(-v^*(m+k-1))=0$. Therefore, we have

\begin{corollary}
For $k\leq 6$, the $k$-SUM problem with $N$ large integers, each with $\Omega(\log N)$ decimal digits, cannot be solved in $O(N^{k-\epsilon})$ time unless the SETH fails.
\end{corollary}

\subsection{A simple approximation}    

Even though we cannot solve the one-bridge problem in $o(n^2)$ time (neither we could for the optimal bridge problem so far), a subquadratic approximation algorithm for the optimal bridge problem can be easily designed as follows.
\begin{enumerate}
    \item Compute the closest pair $(p,q)$, where $p\in V(T_1), q\in V(T_2)$.
    \item Return $\max_{x\in V(T_1)}\delta_{T_1}(x,p)+|pq|+\max_{y\in V(T_2)}\delta_{T_2}(q,y)$ and the corresponding path
    between $x$ and $y$.
\end{enumerate}
The closest pair between $V(T_1)$ and $V(T_2)$ can be computed in $O(n\log n)$ time using a standard method, say, Voronoi diagrams. Note that $\delta_{T_1}(x,p)$
(resp. $\delta_{T_2}(q,y)$) can be computed in linear time by running BFS starting at $p$ (resp. $q$) on the tree $T_1$ (resp. $T_2$); i.e., we do not need the $O(n^2)$ time preprocessing as for the exact algorithm.

\begin{theorem}
The optimal bridge problem can be solved in $O(n^2)$ time
and can be approximated in $O(n\log n)$ time with a factor of 2 (and the factor is tight).
\end{theorem}

\begin{proof}
 The running times are straightforward. For the approximation
 factor, if ${OPT}$ is the optimal solution value, $(p^{*},q^{*})$ is the optimal bridge and $(p,q)$ is the approximate solution computed. Also, let $x^*\in V(T_1)$ and $y^*\in V(T_2)$ be chosen such that $\delta_{T_1}(x,p^*)+|p^*q^*|+\delta_{T_2}(q^*,y)$ is maximized. Then,
 $$\max_{x\in V(T_1)}\delta_{T_1}(x,p)\leq D(V(T_1))\,,$$ where $D(V(T_1))$ is the diameter of the tree~$T_1$.
 Similarly, $\max_{y\in V(T_2)}\delta_{T_2}(q,y)\leq D(V(T_2))$.
 On the other hand, let $\text{OPT}=\delta_{T_1}(x^*,p^*)+|p^*q^*|+\delta_{T_2}(q^*,y^*)$.
 Then, $\delta_{T_1}(x^*,p^*)\geq D(V(T_1))/2$ and $\delta_{T_2}(q^*,y^*)\geq D(V(T_2))/2$.
 Consequently, the approximation solution $\text{APP}$ satisfies

 \begin{align*}
    \text{APP} & = \max_{x\in V(T_1)}\delta_{T_1}(x,p) +|pq|+\max_{y\in V(T_2)}\delta_{T_2}(q,y) \\
    &\leq D(V(T_1))+|pq|+D(V(T_2))\\
     &\leq 2\delta_{T_1}(x^*,p^*)+|p^*q^*|+2\delta_{T_2}(q^*,y^*)\\ 
        & \leq 2\cdot\text{OPT}.
    \end{align*}
\end{proof}

To see the tightness of the approximation factor, we refer to
Figure~\ref{fig:fig1}. In this figure, the approximate bridge
is $(c,f)$, giving a solution value of $4n+1-\varepsilon$. The optimal bridge is $(b,e)$, resulting a solution value of $2n+1$. Clearly,
$$\lim_{n\to\infty}\frac{4n+1-\varepsilon}{2n+1}=2.$$ 

\begin{figure}
\psfrag{1}{$1$}
\psfrag{n/2}{$n/2$}
\psfrag{1-e}{$1-\varepsilon$}
\centering
\includegraphics[width=0.18\textwidth]{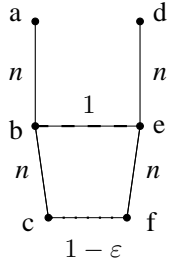}
\caption{\label{fig:fig1}An example where two paths $(a,b,c)$ and $(d,e,f)$ need to be connected into a spanning tree. $be$ and $cf$ are positioned by rotating a segment of length $n$ at $b$ and $e$ slightly to achieve that $|cf|=1-\epsilon$. Adding $(b,e)$ would result in an optimal diameter of $2n+1$; with the greedy method of adding the shortest edge between the two paths, $(c,f)$ is added to achieve a diameter of $4n+1-\varepsilon$. (The distances in the figure are not measured geometrically, due to the space constraint.) }
\end{figure}

We note that for the geometric version when two convex polygons are given, the approximation algorithm can make use of the centers of the polygons; hence the factor can be
improved to $\sqrt{2}$ \cite{DBLP:journals/comgeo/AhnCS03}.
Nonetheless, if we need to connect $k+1$ tree road networks
using $k$ edges to minimize the diameter of the resulting tree, combining our approximation with a method to \cite{DBLP:journals/comgeo/AhnCS03}
can be used to obtain a factor of~4.

\section{The twin bridges problem}

In this section we focus on the twin bridges problem on two road networks which are given as disjoint (geometric) trees.
Given two disjoint trees $T_1$ and $T_2$ with a total of $n$
vertices (which are points in the plane), we aim to add pairs
$(p_1,q_1)$ and $(p_2,q_2)$ with $p_i\in V(T_1), q_i\in V(T_2), i\in \{1,2\}$, and $p_1\neq p_2, q_1\neq q_2$, as new edges/bridges into $T_1$ and $T_2$ to
form a road network $T''=T_1\cup T_2\cup\{(p_1,q_1),(p_2,q_2)\}$, the goal is to minimize the maximum distance between nodes of $T''$ 
whose corresponding path must pass through one or two of these bridges. We loosely call this minimum distance the {\em constrained} diameter of
$T''$. See Figure~\ref{fig:fig2} for an example. (Note that since $T_1$ and $T_2$ are initially disconnected, at 
least one edge/bridge between $T_1$ and $T_2$ must be added. Consequently, after the first bridge is added, if we do not require the second edge 
to be a bridge between $T_1$ and $T_2$, then adding the second edge in $T_1$ or $T_2$ is a trivial problem. This case is not included in our 
twin bridges problem.)

First of all, notice that $T''$ is not a tree anymore. The constrained diameter of $T''$ can appear in four cases:
\begin{enumerate}
    \item It is the maximum of $\delta_{T_1}(x,p_1)+|p_1q_1|+\delta_{T_2}(q_1,y)$, with $x\in V(T_1)$ and $y\in V(T_2)$;
    \item It is the maximum of $\delta_{T_1}(x,p_2)+|p_2q_2|+\delta_{T_2}(q_2,y)$, with $x\in V(T_1)$ and $y\in V(T_2)$;
    \item It is the maximum of $\delta_{T_1}
    (x,p_1)+|p_1q_1|+\delta_{T_2}(q_1,q_2)+|q_2p_2|+\delta_{T_1}(p_2,z)$, with $x,z\in V(T_1)$; and
    \item It is the maximum of $\delta_{T_2}
    (y,q_1)+|q_1p_1|+\delta_{T_1}(p_1,p_2)+|p_2q_2|+\delta_{T_2}(q_2,w)$, with $y,w\in V(T_2)$.
\end{enumerate}
\begin{figure}
\psfrag{x}{$x$}
\psfrag{y}{$y$}
\psfrag{z}{$z$}
\psfrag{w}{$w$}
\psfrag{x'}{$x'$}
\psfrag{z'}{$z'$}
\psfrag{p'}{$p'$}
\psfrag{p"}{$p''$}
\psfrag{p1}{$p_1$}
\psfrag{p2}{$p_2$}
\psfrag{q1}{$q_1$}
\psfrag{q2}{$q_2$}
\psfrag{T1}{$T_1$}
\psfrag{T2}{$T_2$}
\centering
\includegraphics[width=0.50\textwidth]{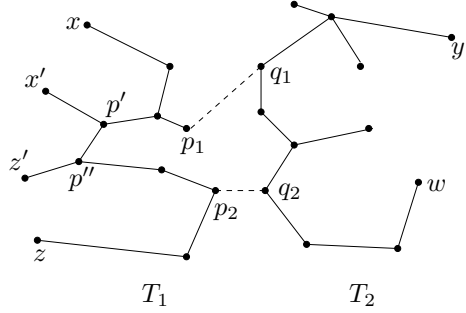}
\caption{\label{fig:fig2}An illustration for the twin bridges problem.}
\end{figure}

Note that Case 1 and Case 2 are similar to the optimal bridge problem in some way; in fact, it is a matter of partitioning $T_i$ into $T_{i,1}$ and $T_{i,2} (i=1,2)$ by deleting an edge in $T_i$ such the two pairs $(T_{i,k}, T_{j,l})$, with $\{i,j\}=\{1,2\}$ and $k,l\in\{1,2\}$, form instances of the optimal bridge problem --- with the resulting optimal
bridges $p_1q_1$ and $p_2q_2$ respectively. (See also the next lemma for the details.) On the other hand, Case 3 and Case 4 are to improve the diameter of $T_1$ (resp. $T_2$)
using both the bridges $p_1q_1$ and $p_2q_2$. (In the optimal bridge problem, it is impossible to improve the diameter of $T_1$ and $T_2$ using the bridge $pq$.
The reason is that $pq$ is a cut edge between $T_1$ and $T_2$, and going from $T_1$ through $pq$ will lead one to $T_2$ or vice versa.)

We refer to Figure~\ref{fig:fig2}, and, assuming a proper
partitioning for $T_1$ and $T_2$ is known as above and the pairing of $(T_{1,k},T_{2,l})$ and $(T_{1,l},T_{2,k})$ is fixed. Then for $x,p_1\in T_{1,k}$, $p_2\in T_{2,l}$, and $y,q_1\in T_{2,l}$ and $q_2\in T_{2,l}$ (or vice versa),
with $\{k,l\}=\{1,2\}$, we define
\begin{eqnarray*}
f_1(x,y,p_1,q_1,p_2,q_2)&=&\delta_{T_{1,k}}(x,p_1)+|p_1q_1|+\delta_{T_{2,l}}(q_1,y)\,,\\
f_2(x,y,p_1,q_1,p_2,q_2)&=&\delta_{T_{1,l}}(x,p_2)+|p_2q_2|+\delta_{T_{2,k}}(q_2,y)\,,\text{ and}\\
f(x,y,p_1,q_1,p_2,q_2)&=&\min\{f_1(x,y,p_1,q_1,p_2,q_2), f_2(x,y,p_1,q_1,p_2,q_2)\}\,.
\end{eqnarray*}

Clearly, to compute $f(-)$, the crucial part is to find
the partition of $T_1$ and $T_2$. We show next that Case 1 and Case 2 can be solved in $O(n^4)$ time, and the following lemma is proved first.

\begin{lemma}
Suppose that the optimal solution for the twin bridges problem occur in Case 1 or Case 2, there must be two edges $e_i\in E(T_i), i\in \{1,2\}$, such that $T_i-\{e_i\}, i\in\{1,2\}$, each results in a pair of trees $(T_{i,1},T_{i,2})$ and the problem is reduced to solving
two optimal bridge problems with input instances $(T_{1,1},T_{2,1})$,
and $(T_{2,1},T_{2,2})$, or vice versa.
\label{lem5}
\end{lemma}

\begin{proof}
We refer to Figure~\ref{fig:fig2}. Suppose that $p_1q_1$ and $p_2q_2$, $p_i\in V(T_1), q_i\in V(T_2)$ for $i\in\{1,2\}$, form the optimal solution for
the twin bridges problem under Case 1 or Case 2. Also suppose that the optimal solution is formed by $x\in V(T_1), y\in V(T_2)$, through the bridge $p_1q_1$; and $z\in V(T_1),w\in V(T_2)$, through the bridge $p_2q_2$. Then in $T_1\cup T_2\cup\{p_1q_1\}\cup\{p_2q_2\}$ there is a cycle $C^*=(p_1,\cdots,p_2,q_2,\cdots,q_1)$, which is 2-connected. Let $e_1=(p',p'')$ be an edge on $C^*$ but not on either of the constrained diameters between $x,y$ and $z,w$, and also on $C^*$, $\delta_{C^*}(p',p_2)\geq\delta_{C^*}(p'',p_1)$ (or $\delta_{T_1}(p'',p_2)\geq\delta_{T_1}(p',p_1)$). Then $(p',p'')$ can be deleted to decompose $T_1$ into $T_{1,1}$ (containing $p_1$) and $T_{1,2}$ (containing $p_2$). The reason we could delete $e_1$ is that for a point $x'$ in $T_{1,1}$ which is closer to $p'$ than $p''$, we have 
\begin{eqnarray*}
\delta_{T_1}(x',p_2)&=&\delta_{T_{1,1}}(x',p')+\delta_{C^*}(p',p_2)\\
&\geq&\delta_{T_{1,1}}(x',p')+\delta_{C^*}(p'',p_1)\\
&=&\delta_{T_{1,1}}(x',p')+|p'p''|+\delta_{T_{1,1}}(p',p_1)\\
&\geq&\delta_{T_{1,1}}(x',p')+\delta_{T_{1,1}}(p',p_1)\\
&=&\delta_{T_{1,1}}(x',p_1).
\end{eqnarray*}
In other words, deleting $e_1$ would not affect the (optimal) diameter from a point in $T_1$ through the bridge $p_1q_1$. Similarly, if $(p',p'')$ is on $C^*$ but not on either of the (optimal) constrained diameters, and we have $\delta_{C^*}(p',p_2)\leq\delta_{C^*}(p'',p_1)$ (or $\delta_{T_1}(p'',p_2)\leq\delta_{T_1}(p',p_1)$, then $(p',p'')$ can be deleted such that for a point $z'$ in
$T_{1,2}$ which is closer to $p''$ than $p'$, we have 
$\delta_{T_1}(z',p_1)\geq\delta_{T_{1,2}}(z',p_2)$.
Symmetrically, $T_2$ can also be decomposed into $T_{2,1},T_{2,2}$ by deleting some edge $e_2\in E(T_2)$.
\end{proof}

With this lemma, it is easy to solve Case 1 and Case 2 in $O(n^4)$ time. We enumerate the $O(n^2)$ pairs of edges
$(e_1,e_2)$ with $e_i\in E(T_i), i\in\{1,2\}$. For each pair,
we delete them from $T_1$ and $T_2$. Then for the two resulting pairs of trees, one from $T_1$ and the other from $T_2$, we solve the optimal bridge problem in $O(n^2)$ time. 

It is noted, surprisingly, that the two optimal bridges $p_1q_1$ and $p_2q_2$ could intersect in some cases, contrary to many geometric problems. In Figure~\ref{fig:fig3}, we show such an example, $\triangle{bcd}$ is initially an equilateral triangle with $da$ being the height of the edge $bc$, we then move $a$ slightly out of the triangle and moving $c$ toward $b$ to have $|bd|>|bc|=|cd|=|ad|$. We also set $|ax|=|by|=|cz|=|dw|=\varepsilon$. If the two intersecting bridges $bc$ and $ad$ are chosen, the optimal solution value is
$|bc|+2\varepsilon$; while if we choose the non-intersecting bridges $ac$ and and $bd$, the solution value is
$|bd|+2\varepsilon$, which is larger. Therefore, when in practical applications the two bridges must not intersect, we might need to seek sub-optimal solutions. 

\begin{figure}
\centering
\includegraphics[width=0.30\textwidth]{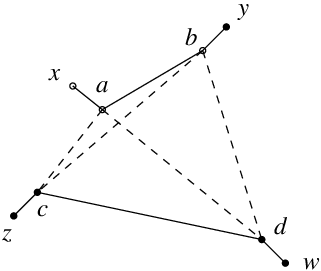}
\caption{\label{fig:fig3}An example for the twin bridges problem in which the two optimal bridges must intersect. The
two input trees are paths: $T_1=(x,a,b,y)$ and $T_2=(z,c,d,w)$.}
\end{figure}

It remains to handle Case 3 and 4 now. We focus on Case 3 here as Case 4 can be similarly handled.
By definition, in Case 3 the optimal solution is the maximum of $\delta_{T_1}(x,p_1)+|p_1q_1|+\delta_{T_2}(q_1,q_2)+|q_2p_2|+\delta_{T_1}(p_2,z)$, with $x,z\in V(T_1)$.
This can be solved as follows. First, let $\delta_{T_1}(x,z)$
be the diameter of $T_1$; naturally, that implies $x,z$ are leaves of $T_1$. Then the problem is to identify two vertices
on $\delta_{T_1}(x,z)$, say $p_1$ and $p_2$, and also two vertices on $T_2$, say $q_1$ and $q_2$, such that  
$$g(p_1,q_1,p_2,q_2)\coloneqq\delta_{T_1}(p_1,p_2)-\{|p_1q_1|+\delta_{T_2}(q_1,q_2)+|q_2p_2|\}$$
is maximized. Obviously, this can be solved by enumerating all
4-tuples $(p_1,q_1,p_2,q_2)$, which incur a cycle
$\langle p_1,q_1,\delta_{T_2}(q_1,q_2),q_2,p_2,\delta_{T_1}(p_2,p_1)\rangle$, where $\delta_{T_1}(p_1,p_2)$ is a subpath on the diameter $\delta_{T_1}(x,z)$. 
Once $(p_1,q_1,p_2,q_2)$ is given, $g(p_1,q_1,p_2,q_2)$ can be computed in $O(1)$ time. Hence all we need to do is to select the maximum of these $O(n^4)$ values
for Case 3 (and symmetrically for Case 4).
By evaluating the recorded values for all recorded solutions in Case 1 and 2 as well, we therefore have the following theorem.

\begin{theorem}
The twin bridges problem can be solved in $O(n^4)$ time.    
\end{theorem}

In the next section, we consider the more general problem
of inserting $k$ edges in a planar road network such that for a set of pairs of vertices we would like to reduce the shortest path distance between each pair. We call this problem Reducing Distances Between Pairs (RDBP), and we show that RDBP is NP-complete. 

\section{RDBP is NP-complete}

We formally define {\em Reducing Distances Between Pairs (RDBP)} as follows: Given a planar straight-line graph $G=(V,E)$ with $n$ vertices and a set of $m$ pairs 
$P=\{(u_i,v_i)\mid u_i,v_i\in V, i\in [m]\}$, insert $k$
edges into $G$ to obtain $G'$ such that the shortest distance between $u_i$ and $v_i$, $i\in [m]$, are all decreased (i.e.,
$\delta_{G'}(u_i,v_i)<\delta_G(u_i,v_i)$ for $i\in [m]$).

Our reduction is from Vertex Cover for Planar 2-Connected Cubic Graphs which is known to be NP-complete \cite{BM01}. If the graph $G$ is not given with a straight-line embedding, then following Fary's theorem \cite{Fary48} it is possible to do that (in fact,
even over a grid of size $\Theta(n^2)$ \cite{DBLP:journals/combinatorica/FraysseixPP90}). Consequently, we can assume that $G$ is given with a straight-line embedding.

The idea of the reduction is as follows. At each vertex $u$
of $G$, for its three incident edges $(u,x)$, $(u,y)$ and $(u,z)$ we cut out a small triangle $\triangle(u_1u_2u_3)$.
Then $u$ is replaced by three paths $\langle u_1,u,u_2,\cdots,x_u\rangle$, $\langle u_2,u,u_1,\cdots,y_u\rangle$ and $\langle u_3,u_1,u,u_2,\cdots,z_u\rangle$.
(All the $\cdots$ represents a subpath of constant size.) Moreover, as all the three paths can be shortened by taking the shortcut $u_1u_2$, if $u$ is selected for a node in a Vertex Cover solution to cover
three edges $(u,x),(u,y)$ and $(u,z)$, then the shortcut $u_1u_2$ must be taken. If we perform this transformation similarly at nodes like $x,y$ and $z$ on $G$ to
obtain $G'$, then the three pairs in $G'$ corresponding to edges $(u,x),(u,y)$ and $(u,z)$ in $G$ would be
$(x_u,u_x)$, $(y_u,u_y)$ and $(z_u,u_z)$ --- they will be stored in the set of pairs $P$, though $u_x,u_y$ and $u_z$ are not drawn in the figure. Note that, by construction, each of these pairs is corresponding to an edge in $G$, which also corresponds to a unique simple (geometric) path, before taking the shortcuts. Moreover, the maximum degree of $G'$ is four.
The construction of $G'$ from $G$ takes $O(n+m)$ time, where $m=O(n)$ and in addition a set $P$ of $m$ pairs are constructed. 
At this point, it is clear that: $G$ has a vertex cover of size $K$ if and only $K$ shortcuts can be taken in $G'$ such that all the shortest distances for pairs $(u,v)$ in $P$ can be reduced.
Note that after the shortcuts are added the resulting graph has a maximum degree of five.
As RDBP is obviously in NP, we have the following theorem.

\begin{figure}
\psfrag{x}{$x$}
\psfrag{y}{$y$}
\psfrag{z}{$z$}
\psfrag{u}{$u$}
\psfrag{xu}{$x_u$}
\psfrag{yu}{$y_u$}
\psfrag{zu}{$z_u$}
\psfrag{u1}{$u_1$}
\psfrag{u2}{$u_2$}
\psfrag{u3}{$u_3$}
\psfrag{G}{$G$}
\centering
\includegraphics[width=0.90\textwidth]{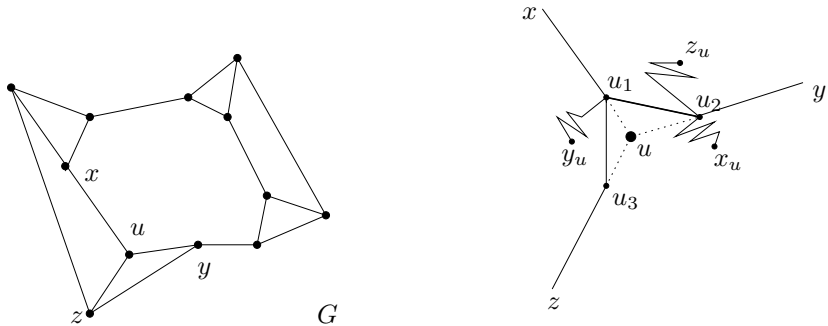}
\caption{\label{fig:fig4}An illustration for the reduction from Vertex Cover on Planar 2-Connected Cubic Graphs to RDBP.}
\end{figure}

\begin{theorem}
The Reducing Distances Between Pairs problem is NP-complete on a given planar graph with degree at most four and the inserted edges have bounded lengths.    
\end{theorem}

\section{Concluding Remarks}

There are several open problems from this paper. For example, for the optimal bridge problem itself, can a near quadratic lower bound be proved?
Also, for RDBP, can we show the NP-completeness of its natural variation, i.e., $P$ is not given and the goal is to minimize the diameter of $G'$?

\section*{Acknowledgments}

This research was carried out when the third author visited University of Trier in September, 2023. We also thank anonymous reviewers for several useful comments.


\end{document}